\documentclass[11pt]{article}

\usepackage[margin=1in]{geometry}
\usepackage{CJK, hyperref, amsmath, amsthm, amssymb, authblk}

\newtheorem{conjecture}{Conjecture}
\newtheorem{lemma}{Lemma}
\newtheorem{observation}{Observation}
\newtheorem{theorem}{Theorem}

\theoremstyle{definition}
\newtheorem{definition}{Definition}

\theoremstyle{remark}
\newtheorem*{remark}{Remark}

\DeclareMathOperator{\poly}{poly}
\DeclareMathOperator{\tr}{tr}
\DeclareMathOperator*{\E}{\mathbb E}
\DeclareMathOperator{\T}{\mathbb T}

\begin{document}

\title{Instability of localization in translation-invariant systems}

\begin{CJK*}{UTF8}{}

\CJKfamily{gbsn}

\author[1]{Yichen Huang (黄溢辰)\thanks{yichuang@mit.edu}}
\author[2]{Aram W. Harrow\thanks{aram@mit.edu}}
\affil[1]{Microsoft Research AI, Redmond, Washington 98052, USA}
\affil[1,2]{Center for Theoretical Physics, Massachusetts Institute of Technology, Cambridge, Massachusetts 02139, USA}

\maketitle

\end{CJK*}

\begin{abstract}

The phenomenon of localization is usually accompanied with the presence of quenched disorder. To what extent disorder is necessary for localization is a well-known open problem. In this paper, we prove the instability of localization in translation-invariant systems. For any translation-invariant local Hamiltonian exhibiting either Anderson or many-body localization, an arbitrarily small translation-invariant random local perturbation almost surely leads to the following manifestations of delocalization: (i) Transport: For any (inhomogeneous) initial state, the spatial distribution of energy or any other local conserved quantity becomes uniform at late times. (ii) Scrambling: The out-of-time-ordered correlator of any traceless local operators decays to zero at late times. (iii) Thermalization: Random product states locally thermalize to the infinite temperature state with overwhelming probability.

\end{abstract}

\section{Introduction}

Thermalization versus localization is a fundamental problem in quantum statistical mechanics. In the presence of quenched disorder, localization can occur not only in single-particle systems, but also in interacting many-body systems. The former is known as Anderson localization (AL) \cite{And58}, and the latter is called many-body localization (MBL) \cite{NH15, AV15, VM16, AP17, AL18, AABS19}.

In the past decade, significant progress has been made towards understanding AL and especially MBL. Characteristic features of localization include, but are not limited to,
\begin{itemize}
\item absence of transport and vanishing dc conductivity \cite{GMP05, BAA06};
\item slow growth of entanglement with time \cite{ZPP08, BPM12, VA13, SPA13U, VA14, ANSS16, Hua17};
\item area law for the entanglement of (almost) all eigenstates \cite{HNO+13, BN13, SPA13, HM14, AS15};
\item (quasi-)local integrals of motion \cite{SPA13, HNO14, RMS15, Imb16, IRS17};
\item intermediate- and late-time behavior of out-of-time-ordered correlators (OTOC) \cite{HZC17, FZSZ17, Che16, SC17, HL17, CZHF17}.
\end{itemize}

As a diagnostic of quantum chaos, OTOC describes information scrambling \cite{LO69, SS14, RSS15, SS15, HQRY16, SBSH16, RY17, ZHC19, Kit14, Kit15}. In MBL systems, it is well known that OTOC of two randomly selected traceless local operators almost surely decays to a non-zero value at late times. (``Almost surely'' means that an event occurs with probability $1$.) This has been confirmed numerically and can be understood from the perspective of (quasi-)local integrals of motion \cite{HZC17, FZSZ17, Che16, SC17, HL17, CZHF17}.

Since localization is usually accompanied with disorder, to what extent disorder is necessary for localization is a well-known open problem. While early works \cite{GF14, DH14A, SM14} suggested the possibility of localization in translation-invariant (TI) systems, it was later argued heuristically that TI may inevitably lead to delocalization \cite{DH14}. Numerical study \cite{SSM15, VLG15, YLC+16} of localization in TI systems appears to be challenging due to significant finite-size effects \cite{PSA15}. More recently, a line of works \cite{Hua15, SKKM17, SKMK17, BDHS18, SKMK18, SKMK19} constructed and studied exactly solvable TI models exhibiting localization, but these models are non-generic in the sense of being highly fine tuned. It was not clear whether their localization properties are robust against generic TI local perturbations.

In this paper, we prove the instability of localization in TI systems from the perspectives of transport, scrambling, and thermalization. For any TI local Hamiltonian exhibiting either AL or MBL, an arbitrarily small TI random local perturbation almost surely leads to the following manifestations of delocalization:
\begin{itemize}
\item Transport: For any (inhomogeneous) initial state, the spatial distribution of energy or any other local conserved quantity becomes uniform at late times.
\item Scrambling: OTOC of any traceless local operators decays to zero at late times.
\item Thermalization: Random product states locally thermalize to the infinite temperature state with overwhelming probability.
\end{itemize}

These results rule out the existence of (quasi-)local integrals of motion and support the viewpoint that quasi-MBL, proposed by Yao {\it et al.} \cite{YLC+16}, is the most localized stable phase in TI systems.

\section{Results}

Throughout this paper, asymptotic notations are used extensively. Let $f,g:\mathbb R^+\to\mathbb R^+$ be two functions. One writes $f(x)=O(g(x))$ if and only if there exist positive numbers $M,x_0$ such that $f(x)\le Mg(x)$ for all $x>x_0$; $f(x)=\Omega(g(x))$ if and only if there exist positive numbers $M,x_0$ such that $f(x)\ge Mg(x)$ for all $x>x_0$; $f(x)=\Theta(g(x))$ if and only if there exist positive numbers $M_1,M_2,x_0$ such that $M_1g(x)\le f(x)\le M_2g(x)$ for all $x>x_0$.

For notational simplicity and without loss of generality, we present the results in one dimension. (It is easy to see that the same results hold in higher spatial dimensions.) Consider a chain of $n$ spins so that the dimension of the Hilbert space is $D=d^n$, where $d=\Theta(1)$ is the local dimension of each spin.

Suppose that $H_{\rm loc}$ is a TI $k$-local Hamiltonian. TI implies periodic boundary conditions, and ``$k$-local'' means that the support of each term in $H_{\rm loc}$ is contained in a consecutive region of size $k$. (For example, a term acting nontrivially only on the first and third spins is $3$- rather than $2$-local.) We assume that $k\ge2$, for any $1$-local Hamiltonian can be artificially regarded as $2$-local. $H_{\rm loc}$ may, but does not have to, exhibit either AL or MBL. Let $h$ be a random Hermitian operator acting on a particular consecutive region of $k$ spins so that the matrix representation of $h$ is of size $d^k\times d^k$. Let $\T$ be the (unitary) lattice translation operator, which acts on the computational basis states as
\begin{equation}
\T(|x_1\rangle\otimes|x_2\rangle\otimes\cdots\otimes|x_n\rangle)=|x_n\rangle\otimes|x_1\rangle\otimes\cdots\otimes|x_{n-1}\rangle,\quad x_l\in\{0,1,\ldots,d-1\},\quad l=1,2,\ldots,n.
\end{equation}
The TI local perturbation is
\begin{equation}
H_{\rm per}=c\sum_{l=0}^{n-1}\T^lh\T^{-l},
\end{equation}
where $c>0$ is an arbitrarily small number. We write the perturbed Hamiltonian as
\begin{equation}
H:=H_{\rm loc}+H_{\rm per},\quad H=\sum_{l=1}^{n}H_l,
\end{equation}
where $H_l=\T^{l-1}H_1\T^{-(l-1)}$ is a local term in $H$.

\begin{definition} [non-degenerate spectrum]
The spectrum of a Hamiltonian is non-degenerate if all eigenvalues are distinct.
\end{definition}

\begin{definition} [non-degenerate gaps]
The spectrum $\{E_j\}$ of a Hamiltonian has non-degenerate gaps if the differences $\{E_j-E_k\}_{j\neq k}$ between different eigenvalues are distinct, i.e., for any $j\neq k$,
\begin{equation}
E_j-E_k=E_{j'}-E_{k'}\implies(j=j')~{\rm and}~(k=k').
\end{equation}
\end{definition}

Since $H_{\rm per}$ is a random perturbation, we expect that the spectrum of $H$ is non-degenerate and has non-degenerate gaps almost surely. Intuitively, for any $k\ge2$, the set of TI $k$-local Hamiltonians whose spectrum is degenerate or has degenerate gaps should be of measure zero. However, proving this statement is not easy. In this direction, Keating {\it et al.} \cite{KLW15} obtained some rigorous results, which we review and extend in Appendix \ref{app}.

Let $\rho(t):=e^{-iHt}\rho e^{iHt}$ be the time evolution of a density operator.

\begin{observation} [transport] \label{trivial}
For any TI Hamiltonian $H=\sum_{l=1}^nH_l$ with non-degenerate spectrum and any (inhomogeneous) initial state $\rho$,
\begin{equation}
\lim_{\tau\to\infty}\frac{1}{\tau}\int_0^\tau\tr(\rho(t)H_l)\,\mathrm dt=\tr(\rho H)/n,\quad\forall l.
\end{equation}
In this sense, every term $H_l$ has the same amount of energy at late times. Furthermore, the same result holds for any other local conserved quantity.
\end{observation}

The late-time OTOC of (not necessarily Hermitian or unitary) operators $A,B$ is given by
\begin{equation} \label{lateOTOC}
{\rm OTOC}^\infty(A,B):=\lim_{\tau\to\infty}\frac{1}{\tau}\int_0^\tau\langle A^\dag(t)B^\dag A(t)B\rangle\,\mathrm dt,
\end{equation}
where $A(t):=e^{iHt}Ae^{-iHt}$ is the time evolution of $A$ in the Heisenberg picture, and $\langle X\rangle:=\tr X/D$ denotes the expectation value of an operator at infinite temperature.

\begin{theorem} [scrambling] \label{main}
For any TI Hamiltonian $H$ whose spectrum has non-degenerate gaps and any traceless local (not necessarily Hermitian or unitary) operators $A,B$ with bounded norm $\|A\|,\|B\|=O(1)$,
\begin{equation} \label{eq:main}
{\rm OTOC}^\infty(A,B)=O(1/n).
\end{equation}
In this sense, OTOC vanishes at late times in the thermodynamic limit $n\to\infty$.
\end{theorem}

\begin{remark}
It is possible to extend this theorem to non-local operators. In particular, the proof of Eq. (\ref{eq:main}) remains valid if $A,B$ are generalized Pauli string operators of length less than $0.999n$.
\end{remark}

\begin{definition} [random product states]
Let $|\psi\rangle=\bigotimes_{l=1}^n|\psi_l\rangle$ be a random product state, where each $|\psi_l\rangle$ is chosen independently and uniformly at random with respect to the Haar measure. Let $\E_{\psi\in P}$ denote the average over the set $P$ of random product states, and $\Pr_{\psi\in P}(\cdots)$ be the probability that an event related to random product states occurs.
\end{definition}

Thermalization of random product states means that local reduced density matrices are almost always close to the infinite temperature state at late times. We prove that this statement holds with overwhelming probability. Indeed, the time average of local reduced density matrices approaches the identity matrix at late times (Theorem \ref{thm:ave}), and the fluctuations of local reduced density matrices are exponentially small in the system size (Theorem \ref{thm:var}).

Since local reduced density matrices are completely determined by the expectation values of local observables, we choose to work with the latter for simplicity. Initialized in the state $|\psi\rangle$, the time-averaged expectation value of a (not necessarily Hermitian) operator $A$ is given by
\begin{equation}
A^\infty_\psi:=\lim_{\tau\to\infty}\frac{1}{\tau}\int_0^\tau\langle\psi|A(t)|\psi\rangle\,\mathrm dt,
\end{equation}
and the fluctuation is given by
\begin{equation}
\Delta A^\infty_\psi:=\lim_{\tau\to\infty}\frac{1}{\tau}\int_0^\tau|\langle\psi|A(t)|\psi\rangle-A^\infty_\psi|^2\,\mathrm dt.
\end{equation}

\begin{theorem} \label{thm:ave}
For any TI Hamiltonian $H$ with non-degenerate spectrum and any traceless local (not necessarily Hermitian) operator $A$ with bounded norm $\|A\|=O(1)$,
\begin{equation} \label{eq:ave}
\E_{\psi\in P}|A^\infty_\psi|=O(1/\sqrt n).
\end{equation}
Furthermore, for any $\delta\ge C\sqrt{(\log n)/n}$ with a sufficiently large constant $C=\Theta(1)$,
\begin{equation} \label{prob}
\Pr_{\psi\in P}(|A^\infty_\psi|\ge\delta)=e^{-\Omega(n\delta^2)}.
\end{equation}
\end{theorem}

\begin{remark}
It is possible to extend Eq. (\ref{eq:ave}) to non-local operators. In particular, the proof of Eq. (\ref{eq:ave}) remains valid if $A$ is a generalized Pauli string operator of length less than $0.999n$.
\end{remark}

\begin{theorem} \label{thm:var}
For any Hamiltonian whose spectrum has non-degenerate gaps and any (not necessarily Hermitian) operator $A$ with bounded norm $\|A\|=O(1)$,
\begin{equation}
\Pr_{\psi\in P}\left(\Delta A^\infty_\psi=e^{-\Omega(n)}\right)=1-e^{-\Omega(n)},
\end{equation}
In this sense, random product states equilibrate exponentially well with overwhelming probability.
\end{theorem}

\begin{remark}
This theorem is very general. It does not assume that the Hamiltonian $H$ is TI, nor any locality of $H$ or $A$.
\end{remark}

We briefly discuss recent related works. Lensky and Qi \cite{LQ19} studied the thermalization of random product states in quantum chaotic systems where the mutual information between subsystems in two copies of the system is small. Theorems \ref{thm:ave}, \ref{thm:var} only assume that the spectrum of the Hamiltonian is non-degenerate and has non-degenerate gaps, respectively, but not any sense of chaoticity.

Using a Berry--Essen theorem for quantum lattice systems \cite{BC15}, Farrelly {\it et al.} \cite{FBC17} proved that any state $|\phi\rangle$ with exponential decay of correlations satisfies that
\begin{equation}
\Delta A^\infty_\phi=\frac{\poly\log n}{s^3\sqrt n},\quad s:=\sqrt{\frac{\langle\phi|H^2|\phi\rangle-\langle\phi|H|\phi\rangle^2}{n}}=O(1).
\end{equation}
Wilming {\it et al.} \cite{WGRE19} proved that
\begin{equation}
\Delta A^\infty_\psi=e^{-\Omega(n)}
\end{equation}
for any product state $|\psi\rangle$ in entanglement-ergodic systems, where the Renyi entanglement entropy of eigenstates in the bulk of the spectrum obeys a volume law for some, not necessarily connected, subsystems. Both Refs. \cite{FBC17, WGRE19} assume that the spectrum of the Hamiltonian has non-degenerate gaps. Theorem \ref{thm:var} is technically neither stronger nor weaker than the results of these references.

\section{Proofs}

Let $\{|1\rangle,|2\rangle,\ldots,|D\rangle\}$ be a complete set of TI eigenstates of $H$ with corresponding energies $E_1\le E_2\le\cdots\le E_D$, and $A_{jk}:=\langle j|A|k\rangle$ be the matrix element of an operator in the energy eigenbasis.

\subsection{Proof of Observation \ref{trivial}}

Writing out the matrix elements,
\begin{multline}
\lim_{\tau\to\infty}\frac{1}{\tau}\int_0^\tau\tr(\rho(t)H_l)\,\mathrm dt=\lim_{\tau\to\infty}\frac{1}{\tau}\int_0^\tau\sum_{j,k}\rho_{jk}(H_l)_{kj}e^{i(E_k-E_j)t}\,\mathrm dt=\sum_{j,k:E_j=E_k}\rho_{jk}(H_l)_{kj}\\
=\sum_j\rho_{jj}(H_l)_{jj}=\sum_j\rho_{jj}E_j/n=\tr(\rho H)/n,
\end{multline}
where we used the assumption of a non-degenerate spectrum and the fact that $(H_l)_{jj}=E_j/n$ for any $l$ due to TI. It is easy to see that the same result holds for any other local conserved quantity.

\subsection{Proof of Theorem \ref{main}}

We use a version of the weak eigenstate thermalization hypothesis (ETH) at infinite temperature, valid for any TI system.

\begin{lemma} [\cite{BKL10, KLW15}] \label{moment}
For any traceless local operator $A$ with bounded norm $\|A\|=O(1)$,
\begin{equation} \label{eq:moment}
\frac{1}{D}\sum_j|A_{jj}|^2=O(1/n).
\end{equation}
\end{lemma}

\begin{proof}
We include a proof for completeness. Let
\begin{equation} \label{TI}
\mathbf A:=\frac{1}{n}\sum_{l=0}^{n-1}\T^lA\T^{-l}
\end{equation}
so that $A_{jj}=\mathbf A_{jj}$ due to TI. Hence,
\begin{equation}
\frac{1}{D}\sum_j|A_{jj}|^2=\frac{1}{D}\sum_j(\mathbf A^\dag)_{jj}\mathbf A_{jj}\le\frac{1}{D}\sum_{j,k}(\mathbf A^\dag)_{jk}\mathbf A_{kj}=\frac{1}{D}\sum_j(\mathbf A^\dag\mathbf A)_{jj}=\langle\mathbf A^\dag\mathbf A\rangle.
\end{equation}
Expanding $\mathbf A$ in the generalized Pauli basis, we count the number of terms that do not vanish upon taking the trace in the expansion of $\mathbf A^\dag\mathbf A$. There are $O(n)$ such terms, each of which is $O(1/n^2)$. Therefore, $\langle\mathbf A^\dag\mathbf A\rangle=O(1/n)$.
\end{proof}

\begin{remark}
It is possible to extend Lemma \ref{moment} to non-local operators. In particular, we still have that $\langle\mathbf A^\dag\mathbf A\rangle=O(1/n)$ and hence the proof of Eq. (\ref{eq:moment}) remains valid if $A$ is a generalized Pauli string operator of length less than $0.999n$.
\end{remark}

The proof of Theorem \ref{main} begins by following the calculations of Ref. \cite{HBZ19}. Writing out the matrix elements,
\begin{equation}
\langle A^\dag(t)B^\dag A(t)B\rangle=\frac{1}{D}\sum_{p,q,r,s}(A^\dag)_{pq}(B^\dag)_{qr}A_{rs}B_{sp}e^{i(E_p-E_q+E_r-E_s)t}.
\end{equation}
Substituting into Eq. (\ref{lateOTOC}),
\begin{equation}
{\rm OTOC}^\infty(A,B)=\frac{1}{D}\sum_{p,q,r,s:E_p+E_r=E_q+E_s}(A^\dag)_{pq}(B^\dag)_{qr}A_{rs}B_{sp}.
\end{equation}
Since the spectrum of $H$ has non-degenerate gaps,
\begin{equation}
E_p+E_r=E_q+E_s\implies((p=q)~{\rm and}~(r=s))~{\rm or}~((p=s)~{\rm and}~(r=q)).
\end{equation}
Hence,
\begin{equation} \label{temp}
{\rm OTOC}^\infty(A,B)=\frac{1}{D}\sum_{j,k}(A^\dag)_{jj}(B^\dag)_{jk}A_{kk}B_{kj}+\frac{1}{D}\sum_{j\neq k}(A^\dag)_{jk}(B^\dag)_{kk}A_{kj}B_{jj}.
\end{equation}
(The constraint $j\neq k$ in the second sum avoids double counting.) The first term on the right-hand side of Eq. (\ref{temp}) is upper bounded by
\begin{multline}
\frac{1}{D}\sum_{j,k}|A_{jj}||A_{kk}||B_{kj}|^2\le\frac{1}{D}\sqrt{\sum_{j,k}|A_{jj}|^2|B_{kj}|^2\times\sum_{j,k}|A_{kk}|^2|B_{kj}|^2}\\
=\frac{1}{D}\sqrt{\sum_j|A_{jj}|^2(B^\dag B)_{jj}\times\sum_k|A_{kk}|^2(BB^\dag)_{kk}}\le\frac{O(1)}{D}\sum_j|A_{jj}|^2=O(1/n),
\end{multline}
where we used the fact that $(B^\dag B)_{jj}\le\|B\|^2=O(1)$ and Lemma \ref{moment}. We complete the proof of Theorem \ref{main} by noting that the second term on the right-hand side of Eq. (\ref{temp}) can be upper bounded similarly.

\subsection{Proof of Theorem \ref{thm:ave}}

Let $c_j:=\langle\psi|j\rangle$ and $p_{\psi j}:=|c_j^2|$. Similar to the calculation in the proof of Observation \ref{trivial},
\begin{multline}
A^\infty_\psi=\lim_{\tau\to\infty}\frac{1}{\tau}\int_0^\tau\langle\psi|A(t)|\psi\rangle\,\mathrm dt=\lim_{\tau\to\infty}\frac{1}{\tau}\int_0^\tau\sum_{j,k}c_jc_k^*A_{jk}e^{i(E_j-E_k)t}\,\mathrm dt=\sum_{j,k:E_j=E_k}c_jc_k^*A_{jk}\\
=\sum_jc_jc_j^*A_{jj}=\sum_jp_{\psi j}A_{jj},
\end{multline}
where we used the assumption of a non-degenerate spectrum. Therefore,
\begin{equation}
\E_{\psi\in P}|A^\infty_\psi|=\E_{\psi\in P}\left|\sum_jp_{\psi j}A_{jj}\right|\le\sum_j|A_{jj}|\E_{\psi\in P}p_{\psi j}=\frac{1}{D}\sum_j|A_{jj}|\le\sqrt{\frac{1}{D}\sum_j|A_{jj}|^2}=\frac{O(1)}{\sqrt n},
\end{equation}
where we used the fact that $\E_{\psi\in P}p_{\psi j}=1/D$ and Lemma \ref{moment}. This completes the proof of Eq. (\ref{eq:ave}).

For the ``furthermore'' part,
\begin{align}
&\Pr_{\psi\in P}(|A^\infty_\psi|\ge\delta)=\int_{\psi\in P:|A^\infty_\psi|\ge\delta}\,\mathrm d\psi\le\frac{2}{\delta}\int_{\psi\in P:|A^\infty_\psi|\ge\delta}\left(|A^\infty_\psi|-\frac{\delta}{2}\right)\,\mathrm d\psi\nonumber\\
&=\frac{2}{\delta}\int_{\psi\in P:|A^\infty_\psi|\ge\delta}\left(\left|\sum_jp_{\psi j}A_{jj}\right|-\frac{\delta}{2}\right)\,\mathrm d\psi\le\frac{2}{\delta}\int_{\psi\in P:|A^\infty_\psi|\ge\delta}\sum_j\left(|A_{jj}|-\frac{\delta}{2}\right)p_{\psi j}\,\mathrm d\psi\nonumber\\
&\le\frac{2}{\delta}\int_{\psi\in P:|A^\infty_\psi|\ge\delta}\sum_{j:|A_{jj}|\ge\delta/2}O(p_{\psi j})\,\mathrm d\psi\le\frac{O(1)}{\delta}\sum_{j:|A_{jj}|\ge\delta/2}\int_{\psi\in P}p_{\psi j}\,\mathrm d\psi=\frac{O(1)}{\delta}\sum_{j:|A_{jj}|\ge\delta/2}\E_{\psi\in P}p_{\psi j}\nonumber\\
&=O(1/\delta)|\{j:|A_{jj}|\ge\delta/2\}|/D,
\end{align}
where we used the normalization conditions $\sum_jp_{\psi j}=1$ and $\E_{\psi\in P}p_{\psi j}=1/D$. Since $O(1/\delta)$ can be absorbed into the notation $e^{-\Omega(n\delta^2)}$ when $\delta\ge C\sqrt{(\log n)/n}$, it suffices to prove Lemma \ref{tail} below.

\begin{lemma} \label{tail}
For any traceless local operator $A$ with bounded norm $\|A\|=O(1)$ and any $\delta\ge C\sqrt{(\log n)/n}$ with a sufficiently large constant $C=\Theta(1)$,
\begin{equation} \label{eq:eth}
|\{j:|A_{jj}|\ge\delta\}|/D=e^{-\Omega(n\delta^2)}.
\end{equation}
\end{lemma}

\begin{proof}
Assume without loss of generality that $A$ is Hermitian. Otherwise, we decompose it as
\begin{equation}
A=\Re A+i\Im A,\quad\Re A:=(A^\dag+A)/2,\quad\Im A:=i(A^\dag-A)/2
\end{equation}
so that both the real part $\Re A$ and the imaginary part $\Im A$ are Hermitian. If we can prove that
\begin{equation}
|\{j:|(\Re A)_{jj}|\ge\delta/2\}|/D=e^{-\Omega(n\delta^2)},\quad |\{j:|(\Im A)_{jj}|\ge\delta/2\}|/D=e^{-\Omega(n\delta^2)},
\end{equation}
then Eq. (\ref{eq:eth}) follows from the union bound.

We define $\mathbf A$ as in Eq. (\ref{TI}) so that $A_{jj}=\mathbf A_{jj}$ due to TI. Let $\{|a_1\rangle,|a_2\rangle,\ldots,|a_D\rangle\}$ be a complete set of TI eigenstates of $\mathbf A$ with corresponding eigenvalues $-O(1)=a_1\le a_2\le\cdots\le a_D=O(1)$. Let $m:=|\{j:A_{jj}\ge\delta\}|$. It is not difficult to see that
\begin{equation} \label{opt}
m\delta\le\sum_{j:A_{jj}\ge\delta}A_{jj}=\sum_{j:A_{jj}\ge\delta}\mathbf A_{jj}\le\sum_{j=D-m+1}^Da_j.
\end{equation}

\begin{lemma} [\cite{Ans16}]
For any $\delta'>0$,
\begin{equation}
|\{j:a_j\ge\delta'\}|/D=e^{-\Omega(n\delta'^2)}.
\end{equation}

\end{lemma}
This lemma implies that
\begin{equation}
a_{D-m'}\le O\left(\sqrt{\log(D/m')/n}\right),\quad\forall m'>0,
\end{equation}
and thus for any $m'\le m$,
\begin{multline}
\sum_{j=D-m+1}^Da_j=\sum_{j=D-m+1}^{D-m'}a_j+\sum_{j=D-m'+1}^Da_j\le(m-m')a_{D-m'}+O(m')\\
\le mO\left(\sqrt{\log(D/m')/n}\right)+O(m').
\end{multline}
Combining this inequality with (\ref{opt}), we obtain that
\begin{equation} \label{ave}
\delta\le O\left(\sqrt{\log(D/m')/n}\right)+O(m'/m).
\end{equation}
Letting $m'=cm\delta$ with a sufficiently small constant $c=\Theta(1)$, the last term on the right-hand side can be upper bounded by $\delta/2$. Hence,
\begin{equation}
\frac{\delta}{2}=O\left(\sqrt{\frac{1}{n}\log\frac{D}{cm\delta}}\right)\implies\frac{m}{D}=\frac{e^{-\Omega(n\delta^2)}}{c\delta}=e^{-\Omega(n\delta^2)},
\end{equation}
where we used $\delta\ge C\sqrt{(\log n)/n}$ to absorb $1/\delta$ into the notation $e^{-\Omega(n\delta^2)}$. We complete the proof of Lemma \ref{tail} by noting that $|\{j:A_{jj}\le-\delta\}|/D$ can be upper bounded similarly.
\end{proof}

\begin{remark}
Lemma \ref{tail} is a probabilistic version of the weak ETH at infinite temperature, valid for any TI system. Previously, a very similar but slightly weaker bound for microcanonical ensembles with exponential decay of correlations was proved by Brand\~ao {\it et al.} \cite{BCSB19} based on the work of Mori \cite{Mor16}.
\end{remark}

\subsection{Proof of Theorem \ref{thm:var}}

Theorem \ref{thm:var} follows by directly combining Lemmas \ref{eff} and \ref{effrand} below.

\begin{lemma} [\cite{Tas98, Rei08, LPSW09, Sho11}] \label{eff}
For any Hamiltonian whose spectrum has non-degenerate gaps and any operator $A$ with bounded norm $\|A\|=O(1)$,
\begin{equation}
\Delta A^\infty_\psi=O(1/D^{\rm eff}_\psi),
\end{equation}
where the effective dimension is defined as
\begin{equation}
1/D^{\rm eff}_\psi=\sum_j|\langle\psi|j\rangle|^4.
\end{equation}
\end{lemma}

\begin{proof}
We include a proof for completeness. Let $c_j:=\langle\psi|j\rangle$. Writing out the matrix elements,
\begin{align}
&\Delta A^\infty_\psi=\lim_{\tau\to\infty}\frac{1}{\tau}\int_0^\tau|\langle\psi|A(t)|\psi\rangle-A^\infty_\psi|^2\,\mathrm dt=\lim_{\tau\to\infty}\frac{1}{\tau}\int_0^\tau\left|\sum_{j\neq k}c_jc_k^*A_{jk}e^{i(E_j-E_k)t}\right|^2\,\mathrm dt\nonumber\\
&=\sum_{j\neq k,j'\neq k'}c_jc_k^*c_{j'}^*c_{k'}A_{jk}(A^\dag)_{k'j'}\lim_{\tau\to\infty}\frac{1}{\tau}\int_0^\tau e^{i(E_j-E_k-E_{j'}+E_{k'})t}\,\mathrm dt=\sum_{j\neq k}|c_j|^2|c_k|^2A_{jk}(A^\dag)_{kj}\nonumber\\
&\le\sqrt{\sum_{j\neq k}|c_j|^4A_{jk}(A^\dag)_{kj}\times\sum_{j\neq k}|c_k|^4(A^\dag)_{kj}A_{jk}}\le\sqrt{\sum_j|c_j|^4(AA^\dag)_{jj}\times\sum_k|c_k|^4(A^\dag A)_{kk}}\nonumber\\
&=\sum_jO(|c_j|^4)=O(1/D^{\rm eff}_\psi),
\end{align}
where we used the assumption that the spectrum of the Hamiltonian has non-degenerate gaps and the fact that $(AA^\dag)_{jj}\le\|A\|^2=O(1)$.
\end{proof}

\begin{remark}
There are improvements and extensions of Lemma \ref{eff} in the literature \cite{SF12, Rei12}. See Ref. \cite{GE16} for a comprehensive review of these developments.
\end{remark}

\begin{lemma} \label{effrand}
Let $\{|b_1\rangle,|b_2\rangle,\ldots,|b_D\rangle\}$ be an arbitrary orthonormal basis of the Hilbert space $(\mathbb C^d)^{\otimes n}$. Then,
\begin{equation}
\E_{\psi\in P}\sum_j|\langle\psi|b_j\rangle|^4\le\left(\frac{2}{d+1}\right)^n=e^{-\Omega(n)}.
\end{equation}
Furthermore,
\begin{equation} \label{mar}
\Pr_{\psi\in P}\left(\sum_j|\langle\psi|b_j\rangle|^4=e^{-\Omega(n)}\right)=1-e^{-\Omega(n)}.
\end{equation}
\end{lemma}

\begin{proof}
Let $|\psi\rangle=\bigotimes_{l=1}^n|\psi_l\rangle$, where each $|\psi_l\rangle$ is an independent Haar-random state. We remind the reader of a well-known result. See, e.g., Ref. \cite{Har13} for its proof.
\begin{lemma}
\begin{equation}
\E_{\psi_l}((|\psi_l\rangle\langle\psi_l|)^{\otimes2})=\frac{2\Pi^{\rm sym}_d}{d(d+1)},
\end{equation}
where $\Pi^{\rm sym}_d=(\Pi^{\rm sym}_d)^2$ is the projection onto the symmetric subspace of $2$ qu\emph{d}its.
\end{lemma}

Using this lemma and the fact that $\Pi^{\rm sym}_d\le I$ (the identity matrix),
\begin{multline}
\E_{\psi\in P}\sum_j|\langle\psi|b_j\rangle|^4=\E_{\psi\in P}\sum_j\langle b_j|^{\otimes2}(|\psi\rangle\langle\psi|)^{\otimes2}|b_j\rangle^{\otimes2}=\sum_j\langle b_j|^{\otimes2}\E_{\psi\in P}((|\psi\rangle\langle\psi|)^{\otimes2})|b_j\rangle^{\otimes2}\\
=\sum_j\langle b_j|^{\otimes2}\left(\bigotimes_{l=1}^n\E_{\psi_l}((|\psi_l\rangle\langle\psi_l|)^{\otimes2})\right)|b_j\rangle^{\otimes2}=\sum_j\langle b_j|^{\otimes2}\left(\frac{2\Pi^{\rm sym}_d}{d(d+1)}\right)^{\otimes n}|b_j\rangle^{\otimes2}\le\left(\frac{2}{d+1}\right)^n.
\end{multline}
The probabilistic bound (\ref{mar}) follows from Markov's inequality.
\end{proof}

\section*{Acknowledgments}

Y.H. would like to thank Fernando G.S.L. Brand\~ao, Xie Chen, and Yong-Liang
Zhang for collaboration on related projects \cite{HZC17, HBZ19}. We are
especially grateful to F.G.S.L.B. for inspiring this work: In March 2017, he
suggested using the weak eigenstate thermalization hypothesis \cite{Mor16} to
study the late-time behavior of out-of-time-ordered correlators in
translation-invariant systems. Y.H. was supported by NSF grant PHY-1818914 and a
Samsung Advanced Institute of Technology Global Research Partnership.  AWH was
funded by NSF grants CCF-1452616, CCF-1729369, PHY-1818914 and ARO contract
W911NF-17-1-0433.

\appendix

\section{Spectrum} \label{app}

In this appendix, we consider the spectrum of a traceless TI $k$-local Hamiltonian $H$. Expanding $H$ in the generalized Pauli basis, we assign to $H_1$ all terms whose Pauli string operators start from the first site so that
\begin{equation}
H=\sum_{l=1}^{n}\T^{l-1}H_1\T^{-(l-1)}.
\end{equation}
The expansion of $H_1$ in the generalized Pauli basis \cite{BK08} has $(d^2-1)$ exactly $1$-local terms and $(d^2-1)^2d^{2j-4}$ exactly $j$-local terms for $j=2,3,\ldots,k$. The coefficients of the expansion parametrize $H_1$ and $H$. Thus, we have defined a parameter space $S$ of dimension
\begin{equation}
d^2-1+\sum_{j=2}^k(d^2-1)^2d^{2j-4}=(d^2-1)d^{2k-2}
\end{equation}
such that points in $S$ are in one-to-one correspondence to traceless TI $k$-local Hamiltonians.

\begin{conjecture} \label{conj1}
For any $k\ge2$, the set of TI $k$-local Hamiltonians with degenerate spectrum is of measure zero.
\end{conjecture}

\begin{conjecture}
For any $k\ge2$, the set of TI $k$-local Hamiltonians whose spectrum has degenerate gaps is of measure zero.
\end{conjecture}

Keating {\it et al.} \cite{KLW15} proved Conjecture \ref{conj1} for $d=2$ and when the system size $n$ is an odd prime. For other values of $d$ or $n$, they proved a weaker result.

\begin{lemma} [\cite{KLW15}]
For any $k$, one of the following must be true:
\begin{itemize}
\item The set of TI $k$-local Hamiltonians with degenerate spectrum is of measure zero.
\item All TI $k$-local Hamiltonians have degenerate spectrum.
\end{itemize}
\end{lemma}

We extend this lemma to the case of non-degenerate gaps.

\begin{lemma} \label{sympoly}
For any $k$, one of the following must be true:
\begin{itemize}
\item The set of TI $k$-local Hamiltonians whose spectrum has degenerate gaps is of measure zero.
\item The spectrum of all TI $k$-local Hamiltonians has degenerate gaps.
\end{itemize}
\end{lemma}

\begin{proof}
Let
\begin{equation}
G:=\prod_{p,q,r,s:\{p,q\}\neq\{r,s\}}(E_p+E_q-E_r-E_s)
\end{equation}
so that $G=0$ if and only if the spectrum of a TI $k$-local Hamiltonian $H$ has degenerate gaps. It is easy to see that $G$ is a symmetric polynomial in $E_1,E_2,\ldots,E_D$. The fundamental theorem of symmetric polynomials implies that $G$ can be expressed as a polynomial in $F_1,F_2,\ldots$, where
\begin{equation}
F_l:=\sum_jE_j^l=\tr(H^l).
\end{equation}
Expanding $H^l$ in the generalized Pauli basis and taking the trace, we see that $F_l$ and hence $G$ are polynomial functions from $S$ (the parameter space) to $\mathbb R$. We complete the proof by noting that the zeros of the multivariable polynomial $G:S\to\mathbb R$ are of measure zero unless $G$ is identically zero.
\end{proof}

\begin{remark}
Since a $2$-local Hamiltonian can be artificially regarded as $k$-local for any $k\ge2$, it suffices to give only one counterexample of a $2$-local TI Hamiltonian to exclude the second possibility in Lemma \ref{sympoly} for any $k\ge2$. However, it is not easy to rigorously verify that the spectrum of a particular TI $2$-local Hamiltonian has non-degenerate gaps.
\end{remark}

\bibliographystyle{abbrv}
\bibliography{instability}

\end{document}